\documentclass[prl,twocolumn,showpacs,preprintnumbers,amsmath,amssymb,superscriptaddress]{revtex4}

\usepackage{bm}
\usepackage{graphicx}
\usepackage{amsmath}
\usepackage{amssymb}
\usepackage{amsbsy}
\usepackage{amsfonts}
\usepackage{amsthm}

\usepackage{color}

\newcommand{\comment}[1]{}
\newcommand{\tr}{{\rm Tr}}
\newcommand{\ii}{{\rm i}}
\newcommand{\nn}{\nonumber\\}
\newcommand{\ket}[1]{| #1 \rangle}
\newcommand{\bra}[1]{\langle #1 |}

\newcommand{\matele}[3]{\langle #1 | #2 | #3 \rangle}

\begin{document}

\theoremstyle{plain}
\newtheorem{theorem}{Theorem}
\newtheorem{lemma}[theorem]{Lemma}
\newtheorem{corollary}[theorem]{Corollary}
\newtheorem{conjecture}[theorem]{Conjecture}
\newtheorem{proposition}[theorem]{Proposition}

\theoremstyle{definition}
\newtheorem{definition}{Definition}

\theoremstyle{remark}
\newtheorem*{remark}{Remark}
\newtheorem{example}{Example}

\title{Mutually unbiased measurements  in finite dimensions}

\author{Amir Kalev}
\affiliation{Center for Quantum Information and Control, University of New Mexico, Albuquerque, NM 87131-0001, USA}

\author{Gilad Gour}
\affiliation{Institute for Quantum Science and Technology, University of Calgary, 2500 University Drive NW, Calgary, Alberta, Canada T2N 1N4}

\date{\today}

\begin{abstract}
We generalize the concept of mutually unbiased bases  (MUB) to measurements which are not necessarily described by rank one projectors. As such, these measurements can be a useful tool to study the long standing problem of the existence of MUB. We derive their general form, and show that in a finite, $d$-dimensional Hilbert space, one can construct a complete set of $d+1$ mutually unbiased measurements. Beside of their intrinsic link to MUB, we show, that these measurements' statistics provide complete information about the state of the system. Moreover, they  capture the physical essence of unbiasedness, and in particular, they satisfy non-trivial entropic uncertainty relation similar to $d+1$ MUB. 
\end{abstract}


\maketitle

Unbiased bases are a fundamental concept in the theory quantum kinematics as they are intimately  related to Bohr's  Complementarity Principle \cite{englertbook}. In a finite dimensional Hilbert space, two orthonormal bases, are said to be unbiased if, and only if,  the transition probability from any state of one basis to any state of the second basis is constant, i.e., independent of the chosen states. In a  $d$-dimensional Hilbert space there are at most $d{+}1$ bases which are pairwise unbiased \cite{ivanovic81}. This set is called the set of mutually unbiased bases (MUB).  

MUB are studied in various context in quantum mechanics. They are used in thought experiments such as the so called ``mean king problem''~\cite{engl01,arav03}, they were shown to have interesting connections with symmetric informationally complete positive-operator-valued measures~\cite{woot06}, complex $t$-designs~\cite{klap05,gros07}, and with graph state formalism~\cite{spen13}.  Beyond being of fundamental interest, MUB have practical importance as well. They play an important role in quantum error correction codes~\cite{gott96,cald97}, quantum cryptography for secure quantum key exchange~\cite{brub02,cerf02}, quantum state tomography~\cite{wootters89,adam10},  and more recently in the detection of quantum entanglement~\cite{speng12}. 

It is therefore been of great effort and research interest to construct the complete set of MUB. To date, numerous construction methods of complete sets of MUB are known in power of prime dimensions \cite{ivanovic81,wootters89,tal02,klap04,durt05,durt10,spen13}, and each method provides a useful and different insight for the problem of the existence of MUB. Alas, it is still not known whether complete sets of MUB exist in all finite dimensions.  A strong numerical evidence suggests that they do not exist in dimensions which are not power of prime \cite{grassl04,brie08,brie10,raynal11}. 

In this work we propose a new approach for `unbiasedness', by generalizing this notion from  bases to  measurements. More particularly, we consider unbiased measurements in a $d$-dimensional Hilbert space, such that unbiased bases are a special, limiting, case. In fact, this idea enables us to construct {\it all} possible complete sets of $d+1$ mutually unbiased measurements (MUM) in a $d$-dimensional Hilbert space. Therefore, if a complete set of MUB exists, it must be a particular case of this construction. Naturally, this generalization can be useful to study questions  relevant to MUB. One interesting question, for example, is how close can $d+1$ MUM get to a complete set of MUB in a given dimension, say, in dimension 6? Beside their relevance to MUB, MUM may be of interest of their own. We show that they provide a linear inversion formula for quantum state tomography, and that they abide by entropic uncertainty relation, similar to MUB. 

In order to formulate the notion of MUM, we first briefly recall the measurement formalism in quantum mechanics.
Generally, a measurement in quantum mechanics is described by a set of positive operators (sometimes called measures) $M_j\geq0$ that sum to the identity operator,  ${\sum_j}{M_j}=1$. The probability of the $j$th outcome is given by Born's rule, $\tr(M_j\rho)$, where $\rho$ is the state (statistical operator) of the quantum system. This representation of a measurement is therefore called a probability-operator measurement (POM), or equivalently a positive-operator-valued measure (POVM). 

Clearly a basis of a finite-dimensional Hilbert space defines both a set of states and a measurement. In particular,  consider a set of $d+1$ MUB  in a $d$-dimensional Hilbert space, $\{\ket{\psi_{n}^{(b)}}\}$ where $b=1,2,\ldots,d+1$ labels the basis while $n=1,2,\ldots, d$ labels the vectors within a basis. The set of projectors on the $b$th basis vectors, ${\cal B}^{(b)}=\{B^{(b)}_n=\ket{\psi_{n}^{(b)}}\bra{\psi_{n}^{(b)}}|n=1,2,\ldots, d\}$, form a measurement, $B^{(b)}_n\geq0$, and $\sum_nB^{(b)}_n=1$, $\forall b$, with the defining properties of MUB,
\begin{align}\label{mubPOM}
\tr (B^{(b)}_n)=&1,\nonumber\\
\tr (B^{(b)}_n B^{(b')}_{n'})=&\delta_{n,n'}\delta_{b,b'}+(1-\delta_{b,b'})\frac1{d}.
\end{align}
The $B^{(b)}_n$s can be thus regarded as quantum states and as measurement operators. Therefore, the unbiasedness of two bases could be re-stated as a property of two measurements as follows: In a $d$-dimensional Hilbert space, measurements of two  bases are  unbiased if and only if when the system is in any state of one basis (measurement), the probability distribution upon measuring it in the second measurement, is completely random.

The notion of unbiasedness of measurements of bases can be therefore generalized to general measurements.\\
\begin{definition}
Two measurements on a $d$-dimensional Hilbert space, ${\cal P}^{(b)}=\{P^{(b)}_n|P^{(b)}_n\geq0,\;\sum_{n=1}^{d}P^{(b)}_n=1\}$, $b{=}1,2$, with $d$ elements each,  are said to be \emph{mutually unbiased measurements} (MUM) if, and only if,
\begin{align}\label{muPOM}
\tr (P^{(b)}_n)&=1,\nonumber\\
\tr (P^{(b)}_n P^{(b')}_{n'})&=\delta_{n,n'}\delta_{b,b'}\kappa+(1-\delta_{n,n'})\delta_{b,b'}\frac{1-\kappa}{d-1}\nn&+(1-\delta_{b,b'})\frac1{d}.
\end{align}
\end{definition}
According to this definition, each measurement operator $P^{(b)}_n$ can be also regarded as a quantum state, for which $\tr(P^{(b)}_n P^{(b')}_{n'})=1/d$, $\forall b\neq b'$. Therefore, indeed,  if the system is in any state, say of ${\cal P}^{(1)}$, then the probability distribution of measuring it in a second measurement, say ${\cal P}^{(2)}$, is completely random ($1/d$).

The inner product of two elements within the same measurement depends on the {\it efficiency parameter} $\kappa$. The value of this parameter determines how close the measurements operators to rank one projectors, i.e., to MUB. The latter is obtained for $\kappa=1$. The other extreme is $\kappa=\frac1{d}$, which corresponds to the trivial case where all the measurement operators are equal to the completely mixed state. We therefore conclude that the efficiency parameter satisfies, \mbox{$\frac1{d}<{\kappa}\leq1$}~\cite{footnote2}. 

Before moving on, we note that, within this definition, the purity of the states (that is, of the measurements' operators) $P^{(b)}_n$ is a constant equals to $\kappa$. One may consider more general definitions of MUM in which the purity of the states depends, for example, on the measurement label $b$ and on the outcome label $n$, $\kappa_n^{(b)}$. These definitions, however, result in less symmetric MUM, which are the primer objective of our study. 

The definition  above allows us to construct $d+1$ MUM in a $d$-dimensional Hilbert space. Consider an orthogonal basis, $F_k$, for the space of hermitian, traceless operators acting on a $d$-dimensional Hilbert space. Such a basis is composed of $d^2-1$ operators, $F_k=F_k^\dagger$, satisfy $\tr(F_k)=0$,  and $\tr(F_k F_l)=\delta_{k,l}$.  We arrange the basis elements on a grid of $d-1$ columns and $d+1$ rows,
\begin{align}
\begin{matrix}
F_1 & F_2 & \cdots &F_{d-1} \\
F_d & F_{d+1} & \cdots &F_{2(d-1)} \\
\vdots & \vdots & \vdots &\vdots \\
\cdots & \cdots & \cdots &F_{(d+1)(d-1)}.
 \end{matrix}
  \end{align}
It is convenient to re-label the operators by a double index $(n,b)$ according to their (column, row) location $n=1,2,\ldots,d-1$, $b=1,2,\ldots,d+1$,
\begin{align}\label{Fs}
\begin{matrix}
  F_{1,1} & F_{2,1} & \cdots &F_{d-1,1}\\
  F_{1,2} & F_{2,2} & \cdots &F_{d-1,2}\\
  \vdots & \vdots & \vdots &\vdots \\
  F_{1,d+1} & F_{2,d+1} & \cdots &F_{d-1,d+1}.
 \end{matrix}
  \end{align}
Next, we define the $d(d+1)$ operators
\begin{align}\label{Fbv}
F^{(b)}_n=\begin{cases}
   F^{(b)}-(d+\sqrt{d}) F_{n,b}& \text{for } n=1,2,\ldots,d-1 \\
   (1+\sqrt{d}) F^{(b)}       & \text{for } n =d
  \end{cases}
\end{align}
with $b=1,2,\ldots,d+1$, and $F^{(b)}$ being the sum of the basis elements on the $b$th row, i.e., $F^{(b)}=\sum_{n=1}^{d-1}F_{n,b}$. This definition ensures the properties,
 \begin{align}\label{sumFbv}
\tr(F^{(b)}_n F^{(b)}_{n'})&=(1+\sqrt{d})^2[\delta_{nn'}(d-1)-(1-\delta_{nn'})],\nn
\sum_{n=1}^d F^{(b)}_n&=0,
\end{align}
which will be used later. We note by passing that by construction, 
\begin{equation}
\tr(F^{(b)}_n F^{(b')}_{n'})=0, \forall b\neq b', \;\forall n,n'=1,2,\ldots,d.
\end{equation}

\begin{theorem}\label{tm:1}
The operators, 
\begin{equation}\label{PinF}
P^{(b)}_n=\frac1{d}+t F^{(b)}_n,
\end{equation}
with $b=1,2,\ldots,d+1$, $n=1,2,\ldots,d$, and the free parameter $t$ chosen such that $P^{(b)}_n\geq0$, form $d+1$ MUM in a $d$-dimensional Hilbert space, where $b$ labels the measurement and $n$ labels the outcome. Moreover, any complete set of MUM have this form.
\end{theorem}
\begin{proof}
To show that the $P^{(b)}_n$s of Eq.~(\ref{PinF}) indeed form MUM, one can verify that they satisfy the definition of MUM, Eq.~(\ref{muPOM}), with efficiency parameter 
\begin{equation}\label{effpara}
\kappa=\frac1{d}+t^2(1+\sqrt{d})^2(d-1).
\end{equation}
The fact that all MUM have the structure of Eq.~(\ref{PinF}) follows by assuming this structure and showing that indeed the $F_{n,b}$s form  a basis for the space of traceless hermitian operators. 
\end{proof}

The efficiency parameter, given in Eq.~(\ref{effpara}), obtains its maximal value $\kappa=1$ for $t^2(1+\sqrt{d})^2=\frac1{d}$. In this case the MUM are actually MUB.  Clearly, the efficiency parameter is determined by the free parameter $t$, which in turn is set such that $P^{(b)}_n\geq0$. This requirement implies that the range of $t$ is 
\begin{equation}\label{tIneq}
-\frac1{d}\frac1{\lambda_{\rm max}}\leq t\leq\frac1{d}\frac1{|\lambda_{\rm min}|}
\end{equation}
where $\lambda_{\rm min}=\min_b\lambda_{\rm min}^{(b)}$,  $\lambda_{\rm max}=\max_b\lambda_{\rm max}^{(b)}$, and $\lambda_{\rm min}^{(b)}$ and $\lambda_{\rm max}^{(b)}$ are the smallest (negative) and largest (positive) eigenvalues of the operators  $F^{(b)}_n$ of Eq.~(\ref{Fbv}) with $n=1,2,\ldots,d$, respectively. [Since  $\tr (F^{(b)}_n)=0$, these operators must have both negative and positive eigenvalues.] The larger the $t$, in its magnitude, the larger the efficiency parameter would eventually be.  Therefore we define the optimal  $t$ to be
\begin{equation}\label{tbopt}
t_{\rm opt}=%
\begin{cases}
\frac1{d}\frac1{|\lambda_{\rm min}|}& \text{for } |\lambda_{\rm min}|<\lambda_{\rm max}, \\
-\frac1{d}\frac1{\lambda_{\rm max}}& \text{for } \lambda_{\rm max}<|\lambda_{\rm min}|.
\end{cases}
\end{equation}
The optimal efficiency parameter of a MUM is given by $\kappa_{\rm opt}=\frac1{d}+t_{\rm opt}^2(1+\sqrt{d})^2(d-1)$. This choice of    $t_{\rm opt}$ ensures  $P^{(b)}_n\geq0$ $\forall b$ and $n$.
The  value of $t_{\rm opt}$ depends on the particular choice of the operator basis $F^{(b)}_n$. For a given choice,
$\kappa_{\rm opt}$ sets an upper bound on how close we can get to MUB. The question whether a complete set of MUB exists in any finite dimension is then translated to the question whether there exists an operator basis for which $\kappa_{\rm opt}=1$.  In the Supplementary Information section we show that for the case where the traceless hermitian operator basis is the generalized Gell-Mann operator basis, an optimal efficiency parameter can be analytically calculated, $\kappa_{\rm opt}=\frac1{d}+\frac2{d^2}$. In a sense, the Gell-Mann operator basis is not a  good choice for a basis, since the optimal efficiency parameter very close to its minimal value $\frac1{d}$ (except for $d=2$, for which it equals 1). Of course, one may resort to numerical methods to estimate better value of $\kappa_{\rm opt}$ in a given dimension. 

The MUM, share various properties with MUB. For example, the set of the $d+1$ MUM of Eq.~(\ref{muPOM}) are informationally complete, that is, any state of the system is determined completely by the MUM outcomes' probabilities, $p^{(b)}_n=\tr(P^{(b)}_n\rho)$. In fact, much like MUB, the MUM provide a linear inversion relation,
\begin{equation}\label{rho in terms of R}
\rho=\sum_{n,b}p^{(b)}_nR^{(b)}_n,
\end{equation}
where the reconstruction operators $R^{(b)}_n$ associated with the MUM are linear function of the outcomes $P^{(b)}_n$, 
\begin{equation}
R^{(b)}_n=\frac{d-1}{\kappa d-1}\left(P^{(b)}_n-\frac{d-\kappa}{d^2-1}\right).
\end{equation}
The reconstruction operators satisfy their defining property,
\begin{equation}
\tr(\rho P^{(b)}_{n})=\sum_{n',b'}p^{(b')}_{n'}\tr(R^{(b')}_{n'} P^{(b)}_{n})=p^{(b)}_n.
\end{equation} 
In the Supplementary Information section we are considering a measurement related to MUM with $d^2$ outcomes which, as well, provides a linear inversion formula for quantum state tomography.

Beyond the mathematical generalization of MUB, the MUM have a  physical significance of their own. As discussed above, the MUM elements, $P^{(b)}_n$ of Eq.~(\ref{PinF}), can  be regarded as quantum states, and the outcome of measuring a state of one measurement in any other MUM, is completely random. We now formalize this aspect of complementarity of the MUM by showing that they satisfy a non-trivial entropic uncertainty relation similar to the one satisfied by MUB. When the latter exist, they satisfy the strong entropic uncertainty relation  \cite{ivanovic92,wehner10},
\begin{equation}\label{mubIneq}
\frac1{d+1}\sum_{b=1}^{d+1}H({\cal B}^{(b)},\rho)\geq\log\frac{d+1}{2},
\end{equation}
where $H({\cal B}^{(b)},\rho)=-\sum_{n=1}^{d}p^{(b)}_n\log p^{(b)}_n$ is the Shannon entropy of the probability distribution, $p^{(b)}_n=\matele{\psi^{(b)}_n}{\rho}{\psi^{(b)}_n}$, associated with measuring the system $\rho$ in MUB measurement ${\cal B}^{(b)}=\{ \ket{\psi^{(b)}_n}\bra{\psi^{(b)}_n}|n=1,2,\ldots,d\}$.

\begin{theorem}\label{tm:2}
The complete set of $d+1$ MUM $\{{\cal P}^{(b)}\}$ of Eq.~(\ref{PinF}) in a $d$-dimensional Hilbert space, satisfies the entropic uncertainty relation,
\begin{equation}\label{thm2}
\frac1{d+1}\sum_{b=1}^{d+1}H({\cal P}^{(b)},\rho)\geq\log\frac{d+1}{1+\kappa},
\end{equation}
where $\kappa$ is the efficiency parameter. 
\end{theorem}
The proof is given in the Supplementary Information section. Note that the inequality of Eq.~(\ref{mubIneq}) abide by the MUB is a particular instance of Eq.~(\ref{thm2}) for the maximal value of the parameter efficiency $\kappa$, $\kappa=1$. The MUM with an optimal efficiency parameter, $\kappa_{\rm opt}$, satisfy Eq.~(\ref{thm2}) with $\kappa=\kappa_{\rm opt}$. 

The inequality $\log\frac{d+1}{1+\kappa}\geq\log\frac{d+1}{2}$ for $\kappa\in(\frac1{d},1]$, apparently indicates that the smaller $\kappa$ the stronger the inequality. This counter intuitive result, can be resolved by noting that for  the minimum value of $\kappa=\frac1{d}$, the measurements operators are the completely mixed state, and as such do not provide any information about the state of the system; hence the uncertainty is the largest. 
This implies that the uncertainty of each measurement must be taken into account before concluding about mutual unbiasedness of the MUM.  Indeed, there is a distinction between the uncertainty relation of Eq.~(\ref{mubIneq}) abide by MUB and the   uncertainty relation of Eq.~(\ref{thm2}) abide by general MUM. This distinction originates in the difference between a measurement of a basis and a general measurement, such as MUM. To be more precise, lets define a state-dependent uncertainty associated with a measurement ${\cal M}$ as, say, the Shannon entropy of its probability distribution, $H({\cal M},\rho)$.
Now suppose that  ${\cal M}$ is a measurement of a basis. In this case there always exist states--- the basis states--- for which the state-dependent uncertainty is zero.  
It is therefore reasonable  to define a state-independent uncertainty associated with a measurement as  the minimum of the state-dependent measurement uncertainty,
\begin{equation}
\Delta{\cal M}=\min_\rho H({\cal M},\rho).
\end{equation}
A measurement has in general an uncertainty $\Delta{\cal M}$ larger than  than zero.
Going back to the uncertainty relations, we conclude that since the MUB are measurements with zero uncertainty each, the entropic uncertainty relation of  Eq.~(\ref{mubIneq}) captures the complementarity, or unbiasedness, aspect of the measurements. In contrast, each one of the MUM are uncertain in general. Therefore, though  Eq.~(\ref{thm2}) provides an entropic uncertainty relation, it includes the `self' uncertainty of each measurement. To  account for the unbiasedness feature of the MUM we must subtract the uncertainty of each measurement,
\begin{equation}
\upsilon=\frac1{d+1}\sum_{b=1}^{d+1}[H({\cal P}^{(b)},\rho)-\Delta{\cal P}^{(b)}],
\end{equation}
so that the proper complementarity entropic uncertainty reads
\begin{equation}
\upsilon\geq\log\frac{d+1}{1+\kappa}-\frac1{d+1}\sum_{b=1}^{d+1}\Delta{\cal P}^{(b)}.
\end{equation}
Techniques, for example, introduced in~\cite{friedland13} may be used to calculate $\Delta{\cal P}^{(b)}$.
We have numerically searched for $\Delta{\cal P}^{(b)}$ for the case where the MUM are constructed by the generalized Gell-Mann operator basis as given in the Supplementary Information section. Searching over $10^6$ states in six-dimensional Hilbert space of random rank, we found that the MUM satisfy a non-trivial complementarity relation, as described in Fig.~\ref{fig:d6} and its caption. 
\begin{figure}[t]
\centering
\includegraphics[width=1\linewidth]{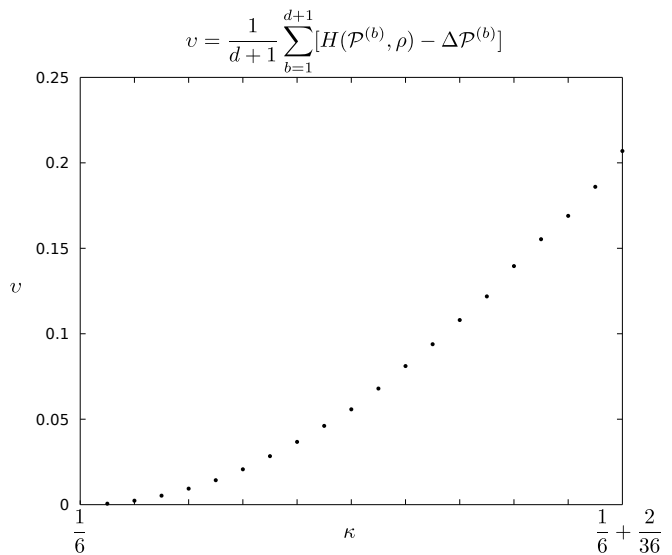}
\caption{Measure of unbiasedness $\upsilon$ as a function of $\kappa$ for MUM on a six-dimensional Hilbert space constructed form the Gell-Mann operator basis described in the Supplementary Information section.  The range of $\kappa$ plotted is  $\frac{1}{d}<\kappa\leq\kappa_{\rm opt}=\frac{1}{d}+\frac{2}{d^2}$. For a given $\kappa$ value, the complementarity on the MUM is larger than  $\upsilon$. }
\label{fig:d6}
\end{figure}%

To conclude, in this work we generalized the notion of unbiasedness of bases, MUB, to general measurements, MUM. We constructed the complete set, i.e., $d+1$ MUM in a $d$-dimensional Hilbert space, where MUB, when exist, are a particular case thereof.   This construction can be used to study, either analytically or numerically,  the problem of the existence of MUB in non power of prime dimensions, and may be helpful to obtain bounds on how close we can get to MUB in these cases. We showed that this mathematical generalization captures the physical essence of unbiased bases in two principle aspects. First, that the probability of obtaining any outcome upon preparing the system in one of measurement elements, and measuring it in the other measurement, is independent of the prepared state and the measurement outcome. And second, a complete set of MUM satisfies a non-trivial entropic uncertainty relation, similar to MUB. Moreover, this relation answers an interesting question regarding the existence of entropic uncertainty relation among $d+1$ measurement settings in finite $d$-dimensional Hilbert space~\cite{wehner10}. 

\emph{Acknowledgments:---}
A.K. research is supported by NSF Grant PHY-1212445. G.G.  research is supported by NSERC.

\begin{titlepage}
\center{\large\textbf{Supplementary Information}\\}
\center{~{ }\\}

\end{titlepage}

\onecolumngrid

\appendix

\section{Optimal efficiency parameter for the case of the generalized Gell-Mann operator basis}
In this appendix we explicitly calculate the optimal efficiency parameter, $\kappa_{\rm opt}$, for the case where the operator basis, $F_{n,b}$, are the well known generalized Gell-Mann (GM) operators (known also as the generalized Pauli operators). For this aim we first consider the range of $t^{(b)}$s making  $P^{(b)}_n\geq0$.

The generalized GM operators are a set of $d^2-1$ operators which form a basis for traceless hermitian operators acting on a $d$-dimensional Hilbert space. We can label them by two indexes $n,m$ each taking on integer values, $n,m=1,2,\ldots,d$, such that
\begin{align}\label{gm1}
G_{n,m}=\begin{cases}
  \frac1{\sqrt2}(\ket{n}\bra{m}+\ket{m}\bra{n})& \text{for } n<m, \\
   \frac\ii{\sqrt2}(\ket{n}\bra{m}-\ket{m}\bra{n})& \text{for } m<n
    \end{cases}
\end{align}
and for $n=m$, with $n=1,2,\ldots,d-1$,
\begin{equation}\label{gm2}
G_{n,n}=  \frac1{\sqrt{n(n+1)}}(\sum_{k=1}^n\ket{k}\bra{k}-n\ket{n+1}\bra{n+1}). 
\end{equation}
The constant in Eqs.~(\ref{gm1}) and~(\ref{gm2}) were chosen such that  $\tr(G_{n,m}^2)=1$ for all possible $n$ and $m$.
It is helpful to think of these operators as elements of a $d\times d$ matrix whose most bottom-right element of the matrix is not `populated' with an operator,
\begin{align}\label{gMat}
\begin{pmatrix}
   G_{1,1} & G_{2,1} & \cdots &G_{d,1} \\
  G_{1,2} & G_{2,2} & \cdots &G_{d,2} \\
  \vdots & \vdots & \vdots &\vdots \\
   G_{1,d} & \cdots & \cdots &[\;]\\
 \end{pmatrix}.
  \end{align}
There is a matter of arbitrariness of how to make the correspondence between the $F$s of Eq.~(\ref{Fs}) and the GM operators. We choose here the following correspondence:
\begin{align}\label{Fg2}
\begin{matrix}
  F_{1,1} & F_{2,1} & \cdots &F_{d-1,1}\\
  F_{1,2} & F_{2,2} & \cdots &F_{d-1,2}\\
  \vdots & \vdots & \vdots &\vdots \\
  F_{1,d+1} & F_{2,d+1} & \cdots &F_{d-1,d+1}
 \end{matrix}\leftrightarrow\begin{matrix}
   G_{2,1} & G_{3,1} & \cdots &G_{d,1} \\
  G_{1,2} & G_{3,2} & \cdots &G_{d,2} \\
  \vdots & \vdots & \vdots &\vdots \\
   G_{1,1} & \cdots & \cdots &G_{d-1,d-1}
 \end{matrix}
  \end{align}
where for the first $d$ rows of Eq.~(\ref{Fs}), we equate the $F$s on the $b$th row  with the $G$s on the $b$th row in the matrix of Eq.~(\ref{gMat}) except the diagonal matrix element. The $F$s on the last, $d+1$th, row correspond to diagonal elements $G_{n,n}$.

Lets now calculate the eigenvalues of $F^{(b)}_n$, or equivalently, the eigenvalues of $F^{(b)}$ and $F^{(b)}-\alpha F_{n,b}$, c.f. Eq.~(\ref{Fbv}). Throughout this section we denote $d+\sqrt{d}\equiv\alpha$, and $1+\sqrt{d}\equiv\beta$. We start with $F^{(1)}$, 
\begin{equation}
F^{(1)}=\frac1{\sqrt2}(\sum_{j=2}^d \ket{1}\bra{j}+\ket{j}\bra{1}),
\end{equation}
whose $d\times d$ matrix representation is
\begin{equation}
F^{(1)}=\frac1{\sqrt2}\begin{pmatrix}
  0 & 1 & \cdots &1 \\
  1 & 0 & \cdots &0 \\
  \vdots & \vdots & \ddots &\vdots \\
    1 & 0 & \cdots &0 \\
 \end{pmatrix}.
\end{equation}
Its eigenvalues $\lambda$ can be found by the eigenvalues equation
\begin{equation}\label{F1det}
{\rm Det}\begin{pmatrix}
  -\lambda & \frac1{\sqrt2} & \cdots &\frac1{\sqrt2} \\
  \frac1{\sqrt2} & -\lambda & \cdots &0 \\
  \vdots & \vdots & \ddots &\vdots \\
    \frac1{\sqrt2} & 0 & \cdots &-\lambda 
 \end{pmatrix}=0.
\end{equation}
To solve this equation we use the identity
\begin{equation}\label{detID}
{\rm Det}\begin{pmatrix}
  A & B  \\
  C & D 
  \end{pmatrix}={\rm Det}(D){\rm Det}(A-BD^{-1}C),
\end{equation}
which holds for invertible $D$. By identifying, 
\begin{align}
 A &=-\lambda, \; B=\frac1{\sqrt2}(1,1,\ldots,1)=C^\dagger,\nn
  D&= \begin{pmatrix}
  -\lambda & \cdots &0 \\
   \vdots & \ddots &\vdots \\
   0 & \cdots &-\lambda 
 \end{pmatrix},
\end{align}
we can write the eigenvalues equation, Eq.~(\ref{F1det}), as
\begin{equation}
\lambda^{d-2}(-2\lambda^2+d-1)=0,
\end{equation}
from which we read the eigenvalues of $F^{(1)}$,
\begin{equation}\label{eigvF1}
\lambda=\begin{cases}
   \pm \sqrt{\frac{d-1}{2}}& \text{multiplicity } 1 \\
   0       & \text{multiplicity }  d-2
  \end{cases}.
\end{equation}
Next we calculate the eigenvalues of $F^{(1)}_1=F^{(1)}-\alpha G_{2,1}$.
The matrix representation of $F^{(1)}_1$ is given by
\begin{equation}
F^{(1)}_1=\begin{pmatrix}
  0 & \frac{1-\alpha}{\sqrt2} &\frac1{\sqrt2}& \cdots &\frac1{\sqrt2} \\
 \frac{1-\alpha}{\sqrt2} & 0 & \cdots & \cdots &0 \\
  \frac1{\sqrt2} & \vdots & \ddots &\ddots &\vdots \\
  \vdots & \vdots & \ddots &\ddots &\vdots\\
    \frac1{\sqrt2} & 0 & \cdots &\cdots&0 \\
 \end{pmatrix},
\end{equation}
and therefore the eigenvalue equation reads
\begin{equation}
{\rm Det}\begin{pmatrix}
  -\lambda & \frac{1-\alpha}{\sqrt2} &\frac1{\sqrt2}& \cdots &\frac1{\sqrt2} \\
  \frac{1-\alpha}{\sqrt2} & -\lambda & \cdots & \cdots &0 \\
  \frac1{\sqrt2} & \vdots & \ddots &\ddots &\vdots \\
  \vdots & \vdots & \ddots &\ddots &\vdots\\
    \frac1{\sqrt2} & 0 & \cdots &\cdots&-\lambda \\
 \end{pmatrix}=0.
\end{equation}
Using the identity of Eq.~(\ref{detID}), with the matrices $A,B,C$ and $D$ defined similarly as before we obtain the polynomial,
\begin{equation}
\lambda^{d-2}(-2\lambda^2+(1-\alpha)^2+d-2)=0.
\end{equation}
from which we read the eigenvalues of $F^{(1)}_1$,
\begin{equation}\label{eigvF11}
\lambda=\begin{cases}
   \pm \sqrt{\frac{(1-\alpha)^2+d-2}{2}}& \text{multiplicity } 1 \\
   0       & \text{multiplicity }  d-2
  \end{cases}.
\end{equation}
These eigenvalues are also the eigenvalues of $F^{(1)}_n$ for $n=1,2,\ldots,d$ since the matrix representation of $F^{(1)}_n$ is the same as the matrix of $F^{(1)}_1$ up to the location of $1-\alpha$ in the first row and column. Actually the eigenvalues of $F^{(1)}$ and $F^{(1)}_1$ are also the eigenvalues of $F^{(b)}$ and $F^{(b)}_n$ for all $b=1,2,\ldots,d$, and $n=1,2,\ldots,d$. This follows by transforming the matrix representation of $F^{(b)}$ (or of $F^{(b)}_n$) to have a similar structure as the matrix of $F^{(1)}$ ($F^{(1)}_n$) by even number of rows and column permutation. Therefore, the eigenvalues equation, and therefore the eigenvalues, of the (permuted) $F^{(b)}$ ($F^{(1)}_n$) are the same as the eigenvalues of $F^{(1)}$ ($F^{(1)}_n$).
Since $\beta\sqrt{\frac{d-1}{2}}=\sqrt{\frac{(1-\alpha)^2+d-2}{2}}$ we conclude that the eigenvalues of $F^{(b)}_n$ of Eq.~(\ref{Fbv}) are $(\pm\beta\sqrt{\frac{d-1}{2}},0)$ for $b=1,2,\ldots,d$, and $n=1,2,\ldots,d$, and therefore,
\begin{equation}
-\frac1{d\beta}\sqrt{\frac2{d-1}}\leq t^{(b)}\leq\frac1{d\beta}\sqrt{\frac2{d-1}}\; \text{for } b=1,2,\ldots,d,\; \forall n.
\end{equation}

We are now considering the last row of the correspondence of Eq.~(\ref{Fg2}). Here the $F^{(b=d+1)}_n$ are all represented as diagonal matrices and therefore it is easy to read their eigenvalues. By noting that the maximal eigenvalue of $F^{(b)}$ equals $\sum_{n=1}^{d-1}\frac1{\sqrt{n(n+1)}}$, its minimal eigenvalue equals $-\sqrt{(d-1)/d}$, and inspecting the various other eigenvalues it is not difficult to show that
\begin{equation}
-\sqrt{\frac{d-1}{d}}-\frac\alpha{\sqrt2}\leq t^{(d+1)}\leq\sum_{n=1}^{d-1}\frac1{\sqrt{n(n+1)}}+\frac\alpha{\sqrt{2}}.
\end{equation}
 Following to the discussion just below Eq.~(\ref{tbopt}), we conclude that $t_{\rm opt}=\frac1{d\beta}\sqrt{\frac2{d-1}}$, and therefore $\kappa_{\rm opt}=\frac1{d}+\frac2{d^2}$.

\section{A measurement with $d^2$ outcomes}
For completeness, let us consider the following a measurement with $d^2$ outcomes which is related to the MUM,  
\begin{align}
\Pi^{(b)}_n&=\frac1{d+1}P^{(b)}_n,\; b=1,2,\ldots,d+1,\; n=1,2,\ldots, d-1\nonumber\\
\Pi_{d^2}&=1-\sum_{b=1}^{d+1}\sum_{n=1}^{d-1}\Pi^{(b)}_n=\frac1{d+1}\sum_{b=1}^{d+1} P^{(b)}_d.
\end{align}
We can regard this measurement as our MUM but collecting all of the last outcomes in each measurement, $P^{(b)}_d$ into one element $\Pi_{d^2}$. The probability of obtaining the outcome $\Pi_{d^2}$ is $\frac1{d+1}\sum_{b=1}^{d+1} p^{(b)}_d$, where $p^{(b)}_d$ is the probability to obtain the outcome $P^{(b)}_d$. This implies that we can rewrite Eq.~(\ref{rho in terms of R}) as 
\begin{align}
\rho=(d+1)\Bigl(\sum_{b=1}^{d+1}\sum_{n=1}^{d-1}\frac{p^{(b)}_n}{d+1}R^{(b)}_n+\sum_{b=1}^{d+1}\frac{p^{(b)}_d}{d+1}R^{(b)}_d\Bigr),
\end{align}
from which we read the reconstruction operators of the $\Pi$s
\begin{align}
\Theta^{(b)}_n&=(d+1)R^{(b)}_n,\; b=1,2,\ldots,d+1,\; n=1,2,\ldots, d-1,\nn
\Theta_{d^2}&=(d+1)\sum_{b=1}^{d+1}R^{(b)}_d.
\end{align}

\section{Proof of Theorem~2}
\begin{proof} 
We first note that the $P^{(b)}_n$s, upon subtracted the completely mixed states, $P^{(b)}_n-\frac1{d}=F^{(b)}_n$, define $d+1$ orthogonal subspaces of operators acting on a $d$-dimensional Hilbert space. Therefore, any quantum state can be written as 
\begin{equation}
\rho=\frac1{d}+\sum_{b=1}^{d+1}\sum_{n=1}^{d}r^{(b)}_nF^{(b)}_n.
\end{equation}
In particular, for pure states $\tr (\rho^2)=1$ which together with the first identity of Eq.~(\ref{sumFbv}) implies
\begin{equation}\label{pure}
\sum_{b=1}^{d+1}\left[d\sum_{n=1}^{d}\left(r^{(b)}_n\right)^2-\left(\sum_{n=1}^{d}r^{(b)}_n\right)^2\right]=\frac{d-1}{d(1+\sqrt{d})^2}.
\end{equation}
Next, consider the MUM ${\cal P}^{(b)}$ of Eq.~(\ref{PinF}). The probability distribution of ${\cal P}^{(b)}$ given a state $\rho$ is 
\begin{equation}\label{pbn}
p^{(b)}_n=\tr(P^{(b)}_n\rho)=\frac1{d}+t(1+\sqrt{d})^2\Bigl(d r^{(b)}_n-\sum_{n'=1}^{d}r^{(b)}_{n'}\Bigr).
\end{equation}
To prove the theorem we use the inequality \cite{ivanovic92,wehner10},
\begin{equation}
H({\cal P}^{(b)},\rho)=-\sum_{n=1}^{d}p^{(b)}_n\log p^{(b)}_n\geq-\log \sum_{n=1}^{d}(p^{(b)}_n)^2,
\end{equation}
which holds for any probability distribution. Summing over $d+1$ MUM and dividing by $d+1$ we obtain,
\begin{align}\label{mainIneq}
\frac1{d+1}\sum_{b=1}^{d+1}H({\cal P}^{(b)},\rho)&\geq-\frac1{d+1}\sum_{b=1}^{d+1}\log \sum_{n=1}^{d}(p^{(b)}_n)^2\nn&\geq-\log \frac1{d+1}\sum_{b=1}^{d+1}\sum_{n=1}^{d}(p^{(b)}_n)^2,
\end{align}
where the second inequality is a consequence of the concavity of the log function.
By using Eqs.~(\ref{sumFbv}),~(\ref{pure}), and~(\ref{pbn}), one can show that for pure states
\begin{equation}
\sum_{b=1}^{d+1}\sum_{n=1}^{d}(p^{(b)}_n)^2=\frac1{d}(d+1)+t^2(1+\sqrt{d})^2(d-1)=1+\kappa.
\end{equation}
Plugging this into Eq.~(\ref{mainIneq}) we arrive at Eq.~(\ref{thm2}).
\end{proof}


\begin{thebibliography}{999}
\bibitem{footnote2} Consider a measurement in a $d$-dimensional Hilbert space, such that the inner product between two measurement elements is given  by $\tr(M_j M_k)=\kappa\delta_{j,k}+(1-\delta_{j,k})\gamma$. The two parameters, $\kappa$ and $\gamma$ depend on each other. Indeed form $\sum_{j,k}\tr(M_j M_k)=d$ we obtain $\gamma=\frac{1-\kappa}{d-1}$. Moreover, $\kappa$ is bounded from below: $1=\tr(M_j)\leq\sqrt{ \tr(M_jM_j) \tr(1)}=\sqrt{\kappa d}$, implies that $1/d\leq \kappa$. Since the largest value of $\kappa$ is 1, we get $1/d<\kappa\leq1$.
\end{thebibliography}
\end{document}